\documentclass[conference]{IEEEtran}
\IEEEoverridecommandlockouts

\usepackage{amsbsy}
\usepackage{floatflt} 

\usepackage{amsmath}
\usepackage{amssymb}
\usepackage{times}
\usepackage{graphicx}
\usepackage{xspace}
\usepackage{paralist} 
\usepackage{setspace} 
\usepackage{xypic}
\xyoption{curve}
\usepackage{latexsym}
\usepackage{theorem}
\usepackage{ifthen}
\usepackage{subfigure}
\usepackage[ruled,vlined]{algorithm2e}
\usepackage{booktabs}
\usepackage{algorithmic}

\newcommand{\mypar}[1]{\vspace{0.03in}\noindent{\bf #1.}}

{\theoremheaderfont{\it} \theorembodyfont{\rmfamily}
\newtheorem{theorem}{Theorem}
\newtheorem{lemma}[theorem]{Lemma}
\newtheorem{definition}[theorem]{Definition}

\newtheorem{corollary}[theorem]{Corollary}

\newtheorem{assumption}[theorem]{Assumption}

}

\begin{document}
\title{Distributed Detection over Time Varying Networks: Large
Deviations Analysis}

\author{Dragana Bajovi\'c, Du$\breve{\mbox{s}}$an Jakoveti\'c, Jo\~ao
Xavier, Bruno Sinopoli and Jos\'e M.~F.~Moura 
\thanks{
This work is partially supported by: the Carnegie Mellon|Portugal
Program under a grant from the Funda\c{c}\~{a}o para a Ci\^encia e Tecnologia (FCT) from Portugal; by FCT grants SIPM
PTDC/EEA-ACR/73749/2006; and by ISR/IST plurianual funding (POSC program,
FEDER). Work of Jos\'e~M.~F.~Moura is partially supported by NSF under grants CCF-1011903 and CCF-1018509, and by AFOSR grant
FA95501010291. Dragana Bajovi\'c and Du$\breve{\mbox{s}}$an Jakoveti\'c hold fellowships from FCT.}
\thanks{D. Bajovi\'{c} and Du$\breve{\mbox{s}}$an Jakoveti\'c are with the
Institute for Systems and Robotics
(ISR), Instituto Superior T\'{e}cnico (IST), Lisbon, Portugal, and with
the Department of Electrical and Computer Engineering, Carnegie Mellon
University, Pittsburgh, PA, USA {\tt\small dbajovic@andrew.cmu.edu,
djakovet@andrew.cmu.edu}}%
\thanks{J. Xavier is with the Institute for Systems and Robotics (ISR),
Instituto Superior T\'{e}cnico (IST), Lisbon, Portugal {\tt\small
jxavier@isr.ist.utl.pt}}
\thanks{B. Sinopoli and Jos\'e M.~F.~Moura are with the Department of
Electrical and Computer
Engineering, Carnegie Mellon University, Pittsburgh, PA, USA {\tt\small
brunos@ece.cmu.edu, moura@ece.cmu.edu}}}%


\maketitle

\begin{abstract}
We apply large deviations theory to study asymptotic performance of \emph{running
consensus} distributed detection in sensor networks. Running consensus is a stochastic approximation type algorithm,
 recently proposed. At each time step $k$, the state at
each sensor is updated
  by a local averaging of the sensor's own state and the states of its neighbors
(consensus) and by
  accounting for the new observations (innovation). We assume Gaussian,
spatially correlated observations. We allow the underlying network be time varying,
   provided that the graph that collects the union of links that are online at least once over a finite time window is connected. This paper shows through large deviations that, under stated assumptions on the network connectivity and
      sensors' observations, the running consensus detection asymptotically
      approaches in performance the optimal centralized detection. That is, the Bayes probability
of detection error (with the running consensus detector) decays exponentially to zero as $k \rightarrow \infty$ at the Chernoff information rate--the best achievable rate of the asymptotically optimal centralized detector.
\end{abstract}

\section{Introduction}

We apply large deviations to study the asymptotic performance of
 distributed detection in sensor networks. Each
 node in the network senses the environment and cooperates locally with its
 neighbors to decide between the two hypothesis, $H_1$ and $H_0$.
 The nodes are connected
 by a generic, time varying network, and there is no fusion center.
  Specifically, we consider distributed detection via running
consensus\footnote{The running consensus algorithm is
  a type of recursive stochastic approximation algorithm, see,
e.g.,~\cite{SoummyaEst}. Reference~\cite{SoummyaEst}
   studies more general stochastic approximation
    type algorithms in the context of distributed estimation. We use the
algorithm in form given
    in~\cite{running-consensus} and will refer to it as running
consensus.} that has been recently proposed
in~\cite{running-consensus}. With running consensus,
   at each time $k$, $N$ nodes update their decision variables by:
   1) incorporating new observation (innovation step); and 2) mixing their
decision variables locally
    with the neighbors (consensus step).

We allow the underlying communication graph be (deterministically)
 time varying; but we assume that the graph that collects all communication links
that are online (at least once) within a finite time window $B$ is connected.
We assume Gaussian, spatially correlated, time--uncorrelated
sensors' observations. Under stated assumptions on the network connectivity and
 the sensors' observations, we show that the running consensus distributed detector
 is asymptotically optimal, as the number of observations $k$
  goes to infinity. That is, the
  running consensus distributed detector asymptotically approaches the
performance of the optimal centralized detector. We apply large deviations
   to study the asymptotic performance of both the (asymptotically)
optimal centralized detector,
    which collects observations from all nodes $i$ at each time $k$, and the
  running consensus detector. For both detectors, the Bayes probability
  of error decays as $e^{-k\mathcal{C}}$,
  where $\mathcal{C}$ is the Chernoff distance
  between the distributions of the $N\times1$ observation vectors under
the two hypothesis, i.e., the Chernoff information.



We now briefly review the existing work on distributed detection.
 Distributed detection has been extensively studied.
  Prior work studies parallel fusion architectures (see,
e.g.,~\cite{Varshney-I,Veraavali,Veraavali-LDP,Tsitsiklis-detection,Poor-II,Moura-saddle-point})
 where all nodes communicate with a fusion node. Also,
 consensus-based detection schemes have been studied (with no fusion node)
in, for
example,~\cite{Moura-detection-consensus,moura-cons-detection,consensus-detection},
  where nodes in the network: 1) collect measurements; and 2)
   \emph{subsequently} run the consensus algorithm to fuse their detection
rules.
   The running consensus distributed detection has been proposed
in~\cite{running-consensus-detection}. Running consensus is different
from classical consensus detection,
    as it incorporates new observations at each time step $k$,
    in real time; thus, unlike classical consensus, no delay
    is introduced from collecting observations to reaching consensus.

We now comment on the differences
between this paper and reference~\cite{running-consensus-detection}, which
also
 studies asymptotic optimality
  of distributed detection via running consensus.
Reference~\cite{running-consensus-detection}
   considers the Neyman-Pearson framework,
   while we adopt the Bayesian framework.
Reference~\cite{running-consensus-detection} considers that, as the
number of observations $k$ grows,
   the distribution means under the two hypothesis become closer and
   closer, at the rate of $1/\sqrt{k}$; consequently, as $k \rightarrow
\infty$, there is an asymptotic,
    non zero, probability of miss, and asymptotic, non zero, probability
     of false alarm. In contrast, we assume that the distributions do not
      change with $k$ (do not approach each other,) and the Bayes
probability of error decays to zero; we then
       examine the rate of decay of the Bayes error probability. Further,
reference~\cite{running-consensus-detection} assumes that
the observations at different sensors are independent identically
        distributed, with generic distribution,
        while we assume Gaussian;
          however, we allow for spatial correlation among observations--a
well-suited assumption, e.g., for
          densely deployed wireless sensor networks (WSNs).
Reference~\cite{running-consensus-detection} studies
          the case where the underlying network is randomly varying; we
consider deterministically time varying network.

\mypar{Paper organization} Section~II reviews the large deviations results
and the Chernoff lemma in hypothesis testing.
Section~III explains data and network models that we assume.
Section~IV introduces the (asymptotically) optimal centralized
detection, as if there was a fusion node and
its detection performance. Section~V shows that the distributed running
consensus
  detector asymptotically approaches in performance the optimal
centralized detector. Finally, section~VI summarizes the paper.

 %
%
%
%
%

%
\section{Background}
%
%
%
 In this section, we briefly review
 standard large deviations analysis for binary hypothesis testing and
  standard asymptotic results (in particular, Chernoff lemma) in binary
hypothesis testing.
  We will later use these results throughout the paper.
\subsection{Binary hypothesis testing problem: Log-likelihood ratio test}
Consider the sequence of independent identically distributed (i.i.d.)
$d$-dimensional random vectors (observations) $y(k)$, $k=1,2,...$, and the
binary hypothesis testing problem of deciding whether
 the probability measure (law) generating $y(k)$ is $\nu_0$ (under
hypothesis $H_0$) or $\nu_1$ (under $H_1$).
  Assume that $\nu_1$ and $\nu_0$ are mutually
 absolutely continuous, distinguishable measures. Based on the
 observations $y(1),...,y(k)$, formally, a decision
 test $T$ is a sequence of maps $T_k: {\mathbb R}^{kd} \rightarrow
\{0,1\}$, $k=1,2,...$, with the interpretation that
$T_k(y(1),...,y(k))=l$ means that $H_l$ is decided, $l=0,1$.
Specifically, consider the log-likelihood ratio (LLR) test to decide
between
  $H_0$ and $H_1$, where $T_k$ is given as follows:
  \begin{eqnarray}
  \label{eqn-llr-test-generic}
  \mathcal{D}(k)&:=&\frac{1}{k} \sum_{j=1}^k \log \frac{d \nu_1}{d
\nu_0}\left(y(j)\right) \\
  T_k &=& \mathcal{I}_{\{\mathcal{D}(k)>\gamma_k\}}.
  \end{eqnarray}
Here $L(k):=\log \frac{d \nu_1}{d \nu_0}\left(y(k)\right)$ is the LLR
(given by the Radon-Nikodym derivative of
 $\nu_1$ with respect to $\nu_0$ evaluated at $y(k)$), $\gamma_k$ is a
chosen threshold, and
 $\mathcal{I}_A$ is the indicator of event $A$. The LLR test with threshold
 $\gamma_k=0$, $\forall k$, is asymptotically optimal in the sense of
Bayes probability of error decay rate, as will be explained in
next subsection (II-B).
\subsection{Log-likelihood ratio test: Large deviations}
This subsection studies large deviations for the LLR decision test with
decision variables $\mathcal{D}(k)$ given in
eqn.~\eqref{eqn-llr-test-generic}. The large deviations
analysis will be very useful in estimating the exponential rate at which
the Bayes probability of error decays and
 in showing the asymptotic optimality of the distributed running consensus
detector. We first give the definition of the large deviations
principle~\cite{DemboZeitouni}.

\begin{definition}[Large deviations principle (LDP)]
\label{definition-ldp}
Consider a sequence of real valued random variables
$\{\Theta(k)\}_{k=1}^\infty:=\{\Theta(k)\}$ and
denote by $\theta_k$ the probability measure of $\Theta(k)$. We say that
the sequence of measures $\{\theta_k\}$ satisfies the LDP with a rate
function
$\mathcal{J}: \mathbb R \rightarrow {\mathbb R} \cup \{+ \infty\}$
if the following holds:
\begin{enumerate}
\item For any closed, measurable set $F \subset \mathbb R$:
\[
\limsup_{k \rightarrow \infty} \frac{1}{k} \log \theta_k\left( F \right)
\leq -\inf_{t \in F} \mathcal{J}(t)
\]
\item For any open, measurable set $G \subset \mathbb R$:
\[
\liminf_{k \rightarrow \infty} \frac{1}{k} \log \theta_k\left( G \right)
\geq -\inf_{t \in G} \mathcal{J}(t).
\]
\end{enumerate}
\end{definition}

It can be shown that the sequence of LLR's $\{L(k)\}$, conditioned on
$H_l$, $l=0,1$, is i.i.d. Denote
by $\mu_k^{(l)}$ the probability measure of $\mathcal{D}(k)$ under
hypothesis $H_l$. Using
Cram\'er's theorem (\cite{DemboZeitouni}), it can be shown that the
sequence of measures $\{\mu_k^{(l)}\}$, $l=0,1$,
 satisfies the LDP with good\footnote{Goodness of rate function is
compactness of its sublevel sets.} rate function:
 \begin{equation}
 \label{eqn-lambda-generic}
 \Lambda^\star_{(l)}(t) = \sup_{\lambda \in \mathbb R}  \left( \lambda t -
\Lambda_{(l)}(\lambda) \right),
 \end{equation}
where $\Lambda_{(l)}(\cdot)$ is the log-moment generating function of
$L(k)$ under hypothesis $H_l$:
\begin{equation}
\Lambda_{(l)}(\lambda) = \log \mathbb E \left[ e^{\lambda L(k)}| H_l \right].
\end{equation}
That is, the rate function $\Lambda_{(l)}^\star(t)$ is the
Fenchel-Legendre (F-L) (\cite{DemboZeitouni}) transform of the log-moment
generating function of $L(k)$ under $H_l$. It can be shown that
$\Lambda_{(1)}^\star(t) = \Lambda_{(0)}^\star(t)-t$. We summarize this
result in the following theorem, e.g.,~\cite{DemboZeitouni}:
\begin{theorem}
\label{theorem-ldp}
The sequence of measures $\{ \mu_k^{(l)} \}$ of $\mathcal{D}(k)$ under
$H_l$ satisfies the LDP with good rate function given by
eqn.~\eqref{eqn-lambda-generic}. \end{theorem}
\subsection{Asymptotic Bayes detection performance: Chernoff lemma}
\label{subsect-chernoff-lemma}
We adopt the Bayes minimum probability
 of error detection. Denote by $P^e(k)$ the Bayes probability of error
after $k$ samples are processed: \begin{equation}
 \label{eqn-bayes-P-e}
P^e(k) = P \left( H_0 \right) \alpha(k) + P \left( H_1 \right) \beta(k),
\end{equation}
where $P \left( H_l \right)$ are the prior probabilities, $\alpha(k):=P
\left(  \mathcal{D}(k)>\gamma_k | H_0\right)$ and $\beta(k):
=P\left(  \mathcal{D}(k)\leq \gamma_k|H_1 \right)$ are, respectively,
the probability of false alarm and the probability of miss,
 and $\gamma_k$ is the test threshold.

We will be interested in the rate at which the Bayes probability of error
decays to zero as the number of observations $k$ goes to infinity. Also,
as auxiliary
 results, we will need the rates at which $\alpha(k)$ and $\beta(k)$ go to
zero as $k \rightarrow \infty$. That is, we will be interested in the
following quantities:
\begin{eqnarray}
\label{eqn-p-e-rate}
&\lim_{k \rightarrow \infty}& \frac{1}{k} \log P^e(k)\\
\label{eqn-alpha-rate}
&\lim_{k \rightarrow \infty}& \frac{1}{k} \log \alpha(k) \\
\label{eqn-beta-rate}
&\lim_{k \rightarrow \infty}& \frac{1}{k} \log \beta(k).
\end{eqnarray}
Theorem~\ref{chernoff-lemma} (\cite{DemboZeitouni}) states that, among all
possible decision tests,
the LLR test with zero threshold minimizes~\eqref{eqn-p-e-rate}. This
result is a corollary of the
Theorem~\ref{theorem-rate-alpha-beta} (\cite{DemboZeitouni}), that asserts
that, for a LLR test with fixed threshold $\gamma_k=\gamma$, $\alpha(k)$
and $\beta(k)$  indeed (simultaneously) decay to zero exponentially;
also, Theorem~3 expresses the exponential rate of decay in terms of the
rate functions defined in
eqns.~\eqref{eqn-alpha-rate}~and~\eqref{eqn-beta-rate}. Before
   stating the Theorem, define $\overline{L}_{(l)}: = \mathbb E \left(
L(k)|H_l \right)$, $l=0,1$.
\begin{theorem}
\label{theorem-rate-alpha-beta}
The LLR test with constant threshold $\gamma_k=\gamma$, $\gamma \in
(\overline{L}_{(0)},
 \overline{L}_{(1)})$ satisfies:
\begin{eqnarray}
\lim_{k \rightarrow \infty} \frac{1}{k} \log \alpha(k) &=& -
\Lambda_{(0)}^\star(\gamma)<0\\
\lim_{k \rightarrow \infty} \frac{1}{k} \log \beta(k)  &=& \gamma -
\Lambda_{(0)}^\star(\gamma)<0.
\end{eqnarray}
\end{theorem}
\begin{theorem}[Chernoff lemma]
\label{chernoff-lemma}
If $P(H_0) \in (0,1)$, then:
\begin{equation}
\label{eqn-theorem-chernoff-lemma}
\inf_{T} \liminf_{k \rightarrow \infty} \left\{ \frac{1}{k} \log P^e(k)
\right\} = -\Lambda_{(0)}^\star(0),
\end{equation}
where the infimum over all possible tests $T$ is attained for the LLR test
with $\gamma_k=0$, $\forall k$.
\end{theorem}
The quantity $\Lambda_{(0)}^\star(0)=\Lambda_{(1)}^\star(0)$ is called the
Chernoff distance between
the distributions of $y(k)$ under $H_0$ and $H_1$, or Chernoff
information,~\cite{DemboZeitouni}.

\mypar{Asymptotically optimal test} We introduce the following definition
of the asymptotically optimal test.
\begin{definition}
\label{def-asym-optimal}
The decision test $T$ is asymptotically optimal if it attains the infimum
in eqn.~\eqref{eqn-theorem-chernoff-lemma}.
\end{definition}
We will show that, for the distributed Gaussian hypothesis testing over
time varying networks,
the running consensus is asymptotically optimal in the sense of
Definition~\ref{def-asym-optimal}.




%
%

\section{Distributed detection model: Data and Network models}
This section describes: 1) the data model
(subsection~\ref{subsect-data-model}), i.e.,
the observation model at each sensor
in the network; and 2) the model of the network through which the sensors
 cooperate with the running consensus distributed detection algorithm
(subsection~III-B). The distributed
 detection algorithm is detailed in Section~V.
\subsection{Data model}
\label{subsect-data-model}
We consider Gaussian binary hypothesis testing in spatially correlated noise.
 The sensors operate (in terms of sensing and
communication) synchronously,
 at discrete time steps $k$. At time $k$, sensor $i$ measures
 (scalar) $y_i(k)$. Collect the sensor measurements in
 a vector
$y(k)=(y_1(k),y_2(k),...,y_N(k))^T$, where $N$ is the total number of
sensors. Nature can be in one
of two possible states: $H_1\--$event occurring (e.g., target present); and
$H_0\--$event not occurring (e.g., target absent.) We assume the following
distribution model
for the vector $y(k)$:
\begin{equation}
\mathrm{under}\,\, H_l:\,\,y(k) = m_l + \zeta(k),\,\,l=0,1,
\end{equation}
where $m_l$ is the (constant) signal under hypothesis $H_l$, and
$\zeta(k)$ is zero mean Gaussian additive noise. We assume that
$\{\zeta(k)\}$
 is an independent identically distributed (i.i.d.) sequence of $N \times 1$
 random vectors with distribution $\zeta(k) \sim \mathcal{N}(0,S)$, where $S$
  is a (positive definite) covariance matrix. Thus, with our model,
  the noise is temporally independent, but can be spatially correlated.
  Spatial correlation should be taken into account due to, for example,
  dense deployment of wireless sensor networks, while it is still
reasonable to assume
  that the observations are independent along time. (Conditioned to
   $H_l$, $\{y(k)\}$ are i.i.d. with
   the distribution $\mathcal{N}(m_l,S)$.)

%
%
%
%
%
%
%
\subsection{Network model and data mixing model}
\label{subsect-netw-data-mixing-model}
We consider distributed detection
via running consensus where each node at a time $k$: 1) measures $y_i(k)$;
2) exchanges its
 current decision variable (denote it by $x_i(k)$) with its neighbors; and
3) performs a weighted average of its own decision variable and
 the neighbors' decision variables. The network
  connectivity is assumed time varying. The weighted averaging, at
each time $k$,
 as with the standard consensus algorithm,
 is described by the $N\times N$ weight matrix $W(k)$. We assume $W(k)$
 is a symmetric, stochastic matrix (it has nonnegative entries and the
rows sum to 1.) The weight
  matrix $W(k)$ respects the sparsity pattern of the network,
   i.e., $W_{ij}(k)=0$, if the
link $\{i,j\}$ is down at time $k$.
We define also the undirected graph $G(k)=(\mathcal{V}, \mathcal{E}(k))$,
where $\mathcal{V}$ is the set of nodes
with cardinality $|\mathcal{V}|=N$, and
$\mathcal{E}(k)$
is the set of undirected edges that are online at time $k$.
 Formally, $\mathcal{E}(k)=\left\{\{i,j\}:\,\, i<j,\, W_{ij}(k)>0
\right\}$. Define also $J:=(1/N)11^T$,
 where $1$ is $N\times 1$ vector with unit entries. We now summarize the
assumptions
 on the matrices $\{W(k)\}$ and the graphs $G(k)$:
\begin{assumption}
\label{assumption-W(k)}
For the sequence of matrices $\{W(k)\}=\{W(k)\}_{k=1}^\infty$, we
assume the following:
\begin{enumerate}
\item $W(k)$ is symmetric and stochastic, $\forall k$.
\item There exists a scalar $W_{\mathrm{min}}\in(0,1)$, such that
\emph{i}) $W_{ii}(k) \geq W_{\mathrm{min}}$, $\forall i$, $\forall k$;
 and \emph{ii}) $\forall k$, $W_{ij}(k) \geq W_{\mathrm{min}}$, if $i \neq j$ and $\{i,j\} \in
\mathcal{E}(k)$.
 \item
 There exists an integer $1 \leq B < +\infty$, such that, $\forall k$,
  the graph $\left( \mathcal{V}, \cup_{l=k+1}^{k+B} \mathcal{E}(l)  \right)$
   is connected.
\end{enumerate}
\end{assumption}
Assumption~\ref{assumption-W(k)}-3) says that nodes should communicate
sufficiently
often (within finite time windows,) such that the network provides
sufficiently fast information flow.

\section{Centralized detection: Bayes optimal test}
We first consider the centralized detection scenario, as if there was a
fusion node
 that collects and processes all sensor observations. The decision
variable $\mathcal{D}(k)$
  and the LLR decision test are given by eqns.~(1)~and~(2), where now,
under the data assumptions in subsection~\ref{subsect-data-model}:
\begin{equation}
\label{eqn_LLR}
L(k) = (m_1-m_0)^T S^{-1} \left(   y(k) - \frac{m_1+m_0}{2} \right)
\end{equation}
Conditioned on either hypothesis $H_1$ and $H_0$,
$L(k)\sim\mathcal{N} \left( m_L^{(l)}, \sigma_L^2 \right)$, where
\begin{eqnarray}
\label{eqn-m-L}
m_L^{(l)} &=&   \frac{(-1)^{l+1}}{2} (m_1-m_0)^T S^{-1} (m_1-m_0)\\
\label{eqn-sigma-L}
\sigma_L^2 &=& (m_1-m_0)^T S^{-1} (m_1-m_0).
\end{eqnarray}
Define the vector $v \in {\mathbb R}^N$ as
\begin{equation}
\label{eqn-def-v}
v:=S^{-1}(m_1-m_0).
\end{equation}
Then, the LLR $L(k)$ can be written as follows:
\begin{equation}
\label{eqn-L(k)-summable}
L(k) = \sum_{i=1}^N v_i \left( y_i(k) - \frac{[m_1]_i+[m_0]_i}{2} \right)
= \sum_{i=1}^N \eta_i(k),
\end{equation}
where $[m_l]_i$ denotes the $i$-th entry of vector $m_l$, $l=0,1$.
Thus, the LLR at time $k$ is separable, i.e., the LLR
is the sum of the terms $\eta_i(k)$ that depend affinely on the individual
observations $y_i(k)$.
 We will exploit this fact in subsection~\ref{subsect-runn-cons} to derive
the distributed, running consensus, detection
 algorithm.
%
%
%
%
%
%

Applying Theorem~\ref{theorem-ldp} to the sequence $\{\mathcal{D}(k)\}$
(under hypothesis $H_l$, $l=0,1$),
we have that the sequence of measures
 of $\mathcal{D}(k)$ satisfies the LDP with good rate function $I_{(l)}:
\mathbb R \rightarrow \mathbb R \cup \{+\infty\}$, which, by evaluating
the log-moment generating function of $L(k)$ in~\eqref{eqn_LLR} and its
F-L transform, can be shown to be:
    \begin{equation}
    \label{eqn-rate-fcns-for-centralized-detection}
    I_{(l)}(t) = \frac{(t-m_L^{(l)})^2}{2 \sigma_L^2}, l=0,1.
    \end{equation}
We state this result as a Corollary~\ref{corollary-cent-det-rate}.
\begin{corollary}
\label{corollary-cent-det-rate}
The sequence $\{\mathcal{D}(k)\}$, under $H_l$, $l=0,1$, satisfies the LDP
 with good rate function $I_{(l)}(\cdot)$, given by
eqn.~\eqref{eqn-rate-fcns-for-centralized-detection}.
\end{corollary}
We remark that Theorem~\ref{chernoff-lemma} also applies to the
detection problem explained in subsection~\ref{subsect-data-model}. Denote by $P^e_{\mathrm{cen}}(k)$
the Bayes probability
of error for the centralized detector (defined in section~IV,) after $k$ samples are processed. Due
to the continuity of the rate functions
in~\eqref{eqn-rate-fcns-for-centralized-detection}, it can be shown that:
 $\liminf_{k \rightarrow \infty} \frac{1}{k} \log P^e_{\mathrm{cen}}(k)=
 \limsup_{k \rightarrow \infty} \frac{1}{k} \log P^e_{\mathrm{cen}}(k) =
\lim_{k \rightarrow \infty} \frac{1}{k} \log P^e_{\mathrm{cen}}(k)$.
  Thus, Theorem~\ref{chernoff-lemma} in this case simplifies to the
following corollary:
  \begin{corollary}(\emph{Chernoff lemma for the optimal centralized detector})
  The LLR test with $\gamma_k=0$, $\forall k$, is asymptotically optimal
in the sense
  of definition~\ref{def-asym-optimal}. Moreover, for the LLR test with
$\gamma_k=0$, $\forall k$, we have:
  \begin{eqnarray}
  \lim_{k \rightarrow \infty} \frac{1}{k} \log P^e_{\mathrm{cen}}(k) = -
I_{(0)}(0)\,\,\,\,\,\,\,\, \\
  =-\frac{1}{8} (m_1-m_0)^\top S^{-1}(m_1-m_0).\nonumber
  \end{eqnarray}
  \end{corollary}

\mypar{Remark} The LLR test with zero
threshold
is optimal also in the finite time $k$ regime, for all $k$, in the sense
that it minimizes the Bayes probability of error, when the prior
probabilities are $P(H_0)=P(H_1)=0.5$. When the prior probabilities are not equal, the LLR
test is also optimal, but the threshold $\gamma_k$ will be different than zero.
\section{Distributed detection algorithm}
\subsection{Distributed detection via running consensus}
\label{subsect-runn-cons}
We now present a distributed detection algorithm via running consensus.
With this detection algorithm, no fusion node is required, and the
underlying network is generic, time varying. The
running consensus is proposed in~\cite{running-consensus}, and it is a
stochastic approximation type of algorithm (see~\cite{SoummyaEst}).
Reference~\cite{running-consensus} studies the case when the observations
of different sensors at a fixed time $k$ are i.i.d. We extend the running consensus detection algorithm to the case
  of spatially correlated Gaussian observations.


With the running consensus distributed detector, each node
$i$ makes local decisions based on its local decision variable $x_i(k)$: If
$x_i(k)>0$, then $H_1$ is accepted; if $x_i(k) \leq 0$, then $H_0$ is
accepted. At each time step $k$, the local decision variable at node $i$
is improved two-fold: 1) by exchanging information with its immediate
neighbors in the network; 2) by incorporating into the decision process
the new local observation $y_i(k)$. Recall the definition of $\eta_i(k)$
in eqn.~\eqref{eqn-L(k)-summable}. Specifically, the update of the local
decision variable at node $i$ is given by the
following equation:
\begin{eqnarray}
\label{eqn-running-cons-sensor-i}
x_i(k+1) = \frac{k}{k+1}  W_{ii}(k)x_i(k)+\\ \frac{k}{k+1} \sum_{j \in
\Omega_i(k)}W_{ij}(k) x_j(k) + \frac{1}{k+1} N \eta_i(k+1), \,\nonumber \\
 k=1,2,...\nonumber
\end{eqnarray}
\begin{equation*}
 x_i(1) = N \eta_i(1)
\end{equation*}

Here $\Omega_i(k)$ is the (time varying) neighborhood of node $i$ at time $k$, and
$ W_{ij}(k)$ are the (time varying) averaging weights, defined together with the
$N\times N$ (time varying) matrices $W(k)=[W_{ij}(k)]$ in subsection~III-B. Let
$x(k) =
(x_1(k),x_2(k),...,x_N(k))^\top$ and
$\eta(k)=(\eta_1(k),...,\eta_N(k))^\top$.
The algorithm in matrix form is given by:
\begin{eqnarray}
\label{eqn_recursive_algorithm}
x(k+1) = \frac{k}{k+1} W(k) x(k) + \frac{1}{k+1} N \eta(k+1),\\
 k=1,2,...\nonumber
\end{eqnarray}
\begin{equation*}
 x(1) = N \eta(1)
\end{equation*}
%
%
Recall the definition of the $N\times 1$ vector $v$ in~\eqref{eqn-def-v}.
The  sequence of $N\times 1$ random vectors $\{\eta(k)\}$, conditioned to
$H_l$, is i.i.d.
 Vector $\eta(k)$ (under hypothesis $H_l$, $l=0,1$) is Gaussian with mean
$m_{\eta}^{(l)}$ and
covariance $S^{\eta}$:
\begin{eqnarray}
\label{eqn_mu_sigma}
m_{\eta}^{(l)} &=& (-1)^{(l+1)} \mathrm{Diag} \left( v
\right)\,\frac{1}{2}(m_1-m_0)\\
S^{\eta} &=& \mathrm{Diag} \left( v\right)  S  \mathrm{Diag} \left( v\right).
\end{eqnarray}
Here $\mathrm{Diag}(v)$ is a diagonal matrix with the diagonal entries
equal to the entries of $v$.

\subsection{Asymptotic optimality of the distributed detection algorithm}
\label{subection-asympt-analysi-running-cons}
In this subsection, we present our main result, which states that
the distributed detection via running consensus asymptotically achieves the
 performance of the optimal centralized detector,
 in the sense that it approaches the exponential error decay rate
 of the (asymptotically) optimal centralized detector.

Denote the probability measure of $x_i(k)$ under hypothesis
$H_l$ with $\chi^{(l)}_{i,k}$.
First, we show that the sequence of measures $\{\chi_{i,k}\}$,
 for all nodes $i$,
 satisfies the LDP with good rate function; the
 rate function for all nodes $i$ is the same, and it is the same as the
 rate function of the optimal centralized detector in
eqn.~\eqref{eqn-rate-fcns-for-centralized-detection}.

%
We prove that the sequence of measures for $\{x_i(k)\}$ (under $H_l$,
$l=0,1$)
 satisfies the LDP using the G{\"a}rtner-Ellis Theorem from large
deviations theory, see~\cite{DemboZeitouni}. We now state
Theorem~\ref{theorem-ldp-decen}.
\begin{theorem}
\label{theorem-ldp-decen}
Let assumption~\ref{assumption-W(k)} hold. The sequence of measures
$\{\chi_{i,k}^{(l)}\}$,
 for all nodes $i$, satisfies the large deviations
 principle with good rate function. The rate function is the same as for
the optimal centralized
  detector and is given by $I_{(l)}(\cdot)$ in
eqn.~\eqref{eqn-rate-fcns-for-centralized-detection}.
 \end{theorem}
Before proving Theorem~\ref{theorem-ldp-decen}, define $\Phi(k,j)$, for
$k>j \geq 1$, as follows:
\begin{equation}
\label{eqn-def-Phi}
\Phi(k,j):=W(k-1)W(k-2)...W(j),
\end{equation}
and remark that the algorithm in eqn.~\eqref{eqn_recursive_algorithm} can be
written as:
\begin{equation}
\label{eqn_x(k)-equation}
x(k) = \frac{N}{k} \sum_{j=1}^{k-1} \Phi(k,j) \eta(j) + \frac{N}{k}
\eta(k),\,\,k=2,3,...
\end{equation}
Next, recall that $J=(1/N) 1 1^\top$, introduce notation:
\begin{equation}
\label{eqn-def-tilde-Phi}
\widetilde{\Phi}(k,j):=\widetilde{W}(k-1) \widetilde{W}(k-2) ...
\widetilde{W}(j), \,\,k > j \geq 1,
\end{equation}
and remark that
\[
\widetilde{\Phi}(k,j) = \Phi(k,j)-J.
\]
To prove Theorem~\ref{theorem-ldp-decen}, we borrow the following result (Lemma~\ref{Lemma-Phi_k_j})
on the matrices $\widetilde{\Phi}(k,j)$ from reference~\cite{Nedic-Lemma}
(Lemma 3.2). First, denote by $[\widetilde{\Phi}(k,j)]_{il}$ the entry in $i$-th row and $l$-th
 column of matrix $\widetilde{\Phi}(k,j)$.
\begin{lemma}
\label{Lemma-Phi_k_j}
Let Assumption~\ref{assumption-W(k)} hold. Then, for the matrices
$\widetilde{\Phi}(k,j)$,
defined by eqn.~\eqref{eqn-def-tilde-Phi}, there holds:
\begin{equation}
\label{eqn-geometric-decay}
\max_{i,l=1,...,N} |\left[\widetilde{\Phi}(k,j)\right]_{il}| \leq \theta
\, \beta^{k-j},
\end{equation}
where $\theta=\left( 1-\frac{W_{\mathrm{min}}}{4N^2} \right)^{-2}$, and
$\beta=\left( 1  -  \frac{ W_{\mathrm{min}} }{ 4N^2}\right)^{1/B}<1.$
\end{lemma}
Lemma~\ref{Lemma-Phi_k_j} says that, under Assumption~\ref{assumption-W(k)},
the size of the matrix $\widetilde{\Phi}(k,j)$ decays geometrically (in
$k-j$)
to zero. This fact will be important in showing
Theorem~\ref{theorem-ldp-decen}.
\begin{proof}[Proof of Theorem~\ref{theorem-ldp-decen}]
Define, for $\mu \in \mathbb R$, the quantity:
\begin{eqnarray}
\Lambda_k^{(l)} \left( \mu \right) &:=& \log
{\mathbb{E}}_l \left[  \mathrm{exp}\,\left(\mu \,x_i(k)\right)
   \right] \\
& = &   \log {\mathbb E}_l \left[  \mathrm{exp}\, \left( \lambda^\top x(k)
\right)
   \right] ,
\end{eqnarray}
where $\lambda=\mu\,e_i$, $\lambda \in {\mathbb R}^N$, and ${\mathbb E}_l
\left[ a\right]:=\mathbb{E} \left[ a|H_l\right]$,
 $l=0,1$. Here $e_i$ denotes the $i$-th column of $N\times N$ identity matrix. We drop the dependence on $i$ in the definition of
$\Lambda_k^{(l)}(\mu)$
  for notation simplicity. Recall the expressions for $m_L^{(l)}$ and
$\sigma_L^2$ in eqns.~\eqref{eqn-m-L}~and~\eqref{eqn-sigma-L}. We will
show, for all $\mu \in \mathbb R$, the following equality:
  \begin{equation}
  \label{eqn-limit}
  \lim_{k \rightarrow \infty} \frac{1}{k} \Lambda_k^{(l)} (k \mu) =
\frac{1}{2} \sigma_L^2 \mu^2 + m_L^{(l)}\,\mu,
  \end{equation}
  Consider the function $\mu \mapsto \frac{1}{2}\sigma_L^2\, \mu^2
  + m_L^{(l)}\, \mu$; this function is essentially smooth, continuous, and
its domain is $\mathbb R$;
   hence, by the G{\"a}rtner-Ellis theorem (\cite{DemboZeitouni},
Theorem~2.3.6),
   $\{\chi_{i,k}^{(l)}\}$ (the sequence of measures of $x_i(k)$ under $H_l$)
    satisfies the LDP. The corresponding rate function equals the F-L
transform of the function
 $\mu \mapsto \frac{1}{2}\sigma_L^2\, \mu^2
  + m_L^{(l)}\, \mu$; and it is easy to show that the F-L transform of
$\mu \mapsto \frac{1}{2}\sigma_L^2\, \mu^2
  + m_L^{(l)}\, \mu$ equals the rate function $I^{(l)}(\cdot)$ given by
eqn.~\eqref{eqn-rate-fcns-for-centralized-detection}.
   Thus, proving Theorem~{9} reduces to showing~\eqref{eqn-limit}. We
thus proceed with showing~\eqref{eqn-limit}. Namely, we have:
\begin{eqnarray}
\label{eqn-term-to-drop}
\frac{1}{k} \Lambda_k^{(l)} (k \,\mu)=\;\;\;\;\;\;\;\;\;\;\;\;\;\;\;
\;\;\;\;\;\;\;\;\;\;\;\;\;\;\;
\;\;\;\;\;\;\;\;\;\;\;\;\;\;\;
\;\;\;\;\;\;\;\;\;\;\;\;\;\;\;
\;\;  \nonumber
\\
\frac{1}{k} \log
 {\mathbb E}_l
\left[
\mathrm{exp}\, \left(  N\,  \lambda^\top \sum_{j=1}^{k-1} {\Phi}(k, j)
\eta(j)  + N\, \lambda^\top \eta(k)      \right)
\right] \;
\nonumber
\end{eqnarray}
\begin{eqnarray}
=\frac{1}{k} \log
  {\mathbb E}_l
\left[
\mathrm{exp}\, \left(   N \lambda^\top \sum_{j=1}^{k-1} {\Phi}(k, j)
\eta(j)        \right)
\right] \;\;\;\;\;\;\;\;\;\;\;\;\;\;\;\;
\nonumber \\
+\,\,\frac{1}{k} \log   {\mathbb E}_l \left[  \mathrm{exp} \left( N
\lambda^\top \eta(k) \right)  \right],
 \;\;\;\;\;\;\;\;\;\;\;\;\;\;\;\;\;\;\;\;\;\;\;\;\;\;\;\;\;\;\;\;\;\;\;\;\;\;
\nonumber
\end{eqnarray}
where the last equality holds because $\eta(k)$ is independent from
$\eta(j)$, $j=1,...,k-1$. We will be interested in computing the limit
$\lim_{k \rightarrow \infty} \frac{1}{k} \Lambda_k^{(l)}(k\,\mu)$, for
all $\mu \in {\mathbb R}$; with this respect,
remark that
\[
\lim_{k \rightarrow \infty} \frac{1}{k} \log   {\mathbb E}_l \left[
\mathrm{exp} \left( N \lambda^\top \eta(k) \right) \right] =0,
\]
for all $\lambda \in {\mathbb R}^N$, because $\eta(k)$ is a Gaussian
random vector
and hence it has finite log-moment generating function at any point $N
\lambda$.

Thus, we have that $\lim_{k \rightarrow \infty} \frac{1}{k} \Lambda_k(k \,
\mu)=
\lim_{k \rightarrow \infty} \frac{1}{k} {\mathcal{L}}_k^{(l)}(k\, \mu)$,
where
\[ \mathcal{L}^{(l)}_k(k \mu)=
 \log
  {\mathbb E}_l
\left[
\mathrm{exp}\, \left(  N\,  \lambda^\top \sum_{j=1}^{k-1} {\Phi}(k, j)
\eta(j)        \right) \right] ,
\]
and we proceed with the computation of ${\mathcal{L}}_k^{(l)}(k\, \mu)$.
The random variables $N \lambda^\top \Phi(k,j) \eta(j)$, $j=1,...,k-1$,
  are independent; moreover, they are Gaussian random variables,
  as linear transformation of the Gaussian variables $\eta(j)$. Recall that
   $m_{\eta}^{(l)}$ and $S^{\eta}$ denote the mean and the covariance of
   $\eta(k)$ under hypothesis $H_l$. Using the independence of $\eta(j)$
and $\eta(s)$, $s \neq j$,
   and using the expression for the moment generating function of
$\eta(j)$, we obtain successively:
\begin{eqnarray}
\frac{1}{k} {\mathcal{L}}_k^{(l)} (k \mu)=
\frac{1}{k} \log
{\mathbb{E}}_l \left[    \mathrm{exp} \, \left(  N\,\lambda^\top
\sum_{j=1}^{k-1} {\Phi}(k, j) \eta(j)  \right)
\right]        \nonumber
\end{eqnarray}
\begin{eqnarray}
=
\frac{1}{k} \log   {\mathbb E}_l  \left[
\Pi_{j=1}^{k-1} \mathrm{exp}  \left( N\,\lambda^\top  \Phi(k,j) \eta(j)
\right)
\right]
\;\;\;\;\;\;
\end{eqnarray}
\begin{eqnarray}
=
\frac{1}{k} \log
\Pi_{j=1}^{k-1}
\mathrm{exp} \,  \left(   N\,\lambda^\top \Phi(k,j)\, m_{\eta}^{(l)}
\right)
\;\;\;\;\;\;\;\;\;\;\;\,
\nonumber
\end{eqnarray}
\begin{eqnarray}
\mathrm{exp} \,  \left(  \frac{1}{2} N^2\lambda^\top \Phi(k,j) S^{\eta}
\Phi(k,j)^\top \lambda \right)
\;\;\;\;\;\;\;\;\;\;
\;\;\;\;\;\;\;
 \nonumber\\
=
\frac{1}{k} \log
\Pi_{j=1}^{k-1}\,\, \mathrm{exp}\, ( N\,\lambda^\top
\left(\widetilde{\Phi}(k,j)+J\right)\,
 m_{\eta}^{(l)})
 \;\;
 \;
 \nonumber \\
\mathrm{exp}\, \left( \frac{1}{2} N^2 \lambda^\top  \left( J + \widetilde{\Phi}(k,j)
\right)
S^{\eta}
\left(  J + \widetilde{\Phi}(k,j)^\top \right)   \lambda   \right). \nonumber
\end{eqnarray}
Denote further:
\begin{eqnarray}
\label{eqn-delta(k)}
\delta(k) &=& \frac{N}{k}    \lambda^\top \sum_{j=1}^{k-1}
\widetilde{\Phi}(k,j) m_{\eta}^{(l)}  \\
&+&  \frac{N^2}{2k}
\lambda^\top  \sum_{j=1}^{k-1} \widetilde{\Phi}(k,j)S^{\eta}
\widetilde{\Phi}(k,j)^\top
   \lambda
 \nonumber
\\
&+& \frac{N^2}{2k} \lambda^\top J S^{\eta} \sum_{j=1}^{k-1}
\widetilde{\Phi}(k,j)^\top \lambda \nonumber \\
&+& \frac{N^2}{2k} \lambda^\top
\sum_{j=1}^{k-1}
\widetilde{\Phi}(k,j) S^{\eta} J \lambda   \nonumber \\
\overline{\Lambda}_{(l)}(\mu,k)&=& \frac{(k-1)N\lambda^\top J m_{\eta}^{(l)}}{k} \nonumber \\
&+& \frac{N^2(k-1)\lambda^\top J S^{\eta} J \lambda}{2k} , \nonumber
\end{eqnarray}
where dependence on $H_l$ is dropped in the definition of $\delta(k)$.
Then, it is easy to see that $\frac{1}{k} \mathcal{L}_k^{(l)}(k
\mu)=\overline{\Lambda}_{(l)}(\mu,k)+\delta(k)$. Also, we have:
 $\lim_{k \rightarrow \infty} \overline{\Lambda}_{(l)}(\mu,k)=N\lambda^\top J m_{\eta}^{(l)}+ \frac{N^2}{2}\lambda^\top J S^{\eta} J \lambda=: \overline{\Lambda}_{(l)}(\mu).$

Recall the expressions for $v$, $m_L^{(l)}$, $\sigma_L^2$, $m_{\eta}^{(l)}$, and $S^{\eta}$ in
eqns.~\eqref{eqn-def-v},~\eqref{eqn-m-L},~\eqref{eqn-sigma-L}, \eqref{eqn_mu_sigma}, (23). We
proceed with the computation of $\overline{\Lambda}_{(l)}(\mu)$:
\begin{eqnarray}
\overline{\Lambda}_{(l)}(\mu) &=&
(-1)^{l+1}\frac{N}{2} \mu \left( m_1-m_0\right)^\top
 \, \mathrm{Diag} (v) \,J\,e_i
 \nonumber
 \\
 &+& \frac{N^2}{2}\mu^2 \,e_i^\top \,J \, S^{\eta} \,J \,e_i \nonumber
 \\
 &=&  (-1)^{l+1} \frac{1}{2}\mu 1^\top \mathrm{Diag}(v) (m_1-m_0) \nonumber
  \\
  &+&\frac{1}{2}\mu^2 1^\top \mathrm{Diag}(v) S \mathrm{Diag}(v) 1 \nonumber
 \\
 &=&  (-1)^{l+1} \frac{1}{2}\mu \left( m_1-m_0 \right)^\top S^{-1} \left(
m_1-m_0 \right) \nonumber
 \\
 &+& \frac{1}{2}
 \mu^2 \left( m_1-m_0 \right)^\top S^{-1} \left( m_1-m_0 \right) \nonumber
\\
 &=&  m_L^{(l)}\,\mu + \frac{1}{2} \sigma_L^2 \mu^2 .\nonumber
\end{eqnarray}
We proceed by showing that $\delta(k) \rightarrow 0$ as $k \rightarrow
\infty$, which implies
the equality in eqn~\eqref{eqn-limit}. Define the quantities
$\overline{S}$, $\overline{m}$,
 and $\overline{b}$, by:
 \begin{eqnarray}
 \overline{S}  &:=& \max_{i,l=1,...,N} |[S^{\eta}]_{il}|\\
 \overline{m}  &:=& \max_{i=1,...,N}   |[m_{\eta}^{(l)}]_{i}|\\
 \overline{b}  &:=& \max_{l=1,...,N}   |[S^{\eta} J e_i]_l|.
 \end{eqnarray}
 Then, it can be shown that $|\delta(k)|$ is bounded as follows:
 \begin{eqnarray}
 \label{eqn-der}
 |\delta(k)|   &\leq& \frac{N^2}{k}  |\mu| \,\overline{m} \sum_{j=1}^{k-1}
\max_{i,l=1,...,N} |[\widetilde{\Phi}(k,j)]_{il}|\\
 &+& \frac{N^4}{2k} \mu^2 \overline{S} \sum_{j=1}^{k-1}
\left(\max_{i,l=1,...,N} |[\widetilde{\Phi}(k,j)]_{il}|\right)^2 \nonumber
\\
 &+& \frac{N^3}{k}  \mu^2 \overline{b} \sum_{j=1}^{k-1} \max_{i,l=1,...,N}
|[\widetilde{\Phi}(k,j)]_{il}|. \nonumber
 \end{eqnarray}
 Applying Lemma~\ref{Lemma-Phi_k_j} to~\eqref{eqn-der}, and using
  the fact that $\beta^{k-j}<\beta^{k-j-1}$, $k>j$, we obtain
successively:
 \begin{eqnarray}
 |\delta(k)|   &\leq& \frac{\theta}{k} \left( N^2 \overline{m} |\mu|  + N^3
\mu^2 \overline{b}    \right)\sum_{j=1}^{k-1} \beta^{k-j-1}\\
 &+& \frac{\theta^2}{k}\,\frac{N^4}{2} \mu^2 \overline{S}
\,\sum_{j=1}^{k-1} \beta^{2(k-j-1)}\\
 &\leq& \frac{\theta}{k} \left( N^2 \overline{m} |\mu|  + N^3 \mu^2
\overline{b}    \right) \frac{1}{1-\beta}\\
 &+& \frac{\theta^2}{k}\,\frac{N^4}{2} \mu^2 \overline{S}
\,\frac{1}{1-\beta^2}.
 \end{eqnarray}
 Letting $k \rightarrow +\infty$, we get that $|\delta(k)| \rightarrow 0$,
and hence,  $\delta(k) \rightarrow 0$, which establishes eqn.~\eqref{eqn-limit}.
\end{proof}

We are now ready to state the main result on asymptotic optimality of the distributed detector (in the sense of
Definition~\ref{def-asym-optimal}.)

\begin{corollary}(\emph{Chernoff lemma for the distributed detector:
Asymptotic optimality})
\label{chernoff-lemma-dist}
The local decision test $T_{k,i}:={\mathcal{I}}_{\{x_i(k)>0\}}$,
$k=1,2,...$, at each node $i$, is asymptotically optimal
in the sense of Definition~\ref{def-asym-optimal}. The corresponding
exponential decay rate of
the Bayes probability of error, at each node $i$, is given by:
\begin{eqnarray}
\label{eqn-decentralized-P_e-rate}
\lim_{k \rightarrow \infty} \frac{1}{k} \log P^e_{i,\,\mathrm{dis}}(k) = -
I_{(0)}(0)\,\,\,\,\,\, \\
  =-\frac{1}{8} (m_1-m_0)S^{-1}(m_1-m_0).\nonumber
\end{eqnarray}
\end{corollary}
\begin{proof} Denote by $\alpha_{i,\mathrm{dis}}(k)$
and $\beta_{i, \mathrm{dis}}(k)$, respectively,
the probability of false alarm and the probability of miss
 for the distributed detector at sensor $i$, i.e.,
 \begin{eqnarray}
 \label{eqn-prob-false-alarm-distributed-detect}
 \label{eqn-alpha-dist}
 \alpha_{i,\mathrm{dis}}(k)&=&
 P\left( x_i(k)>0|H_0\right)
 =\chi_{i,k}^{(0)}\left( (0,+\infty)\right)
  \nonumber\\
 \label{eqn-beta-dist}
 \beta_{i,\mathrm{dis}}(k)&=&P\left( x_i(k)\leq 0 | H_1
\right)=\chi_{i,k}^{(1)}\left( (-\infty,0]\right). \nonumber
 \end{eqnarray}
Consider now only $\alpha_{i,\mathrm{dis}}(k)$ but the same applies to
$\beta_{i,\mathrm{dis}}(k)$. By Theorem~\ref{theorem-ldp-decen}, the
sequence of measures $\chi_{i,k}^{(0)}$
  satisfies the LDP with good rate function $I_{(0)}(\cdot)$ given in
eqn.~\eqref{eqn-rate-fcns-for-centralized-detection}. Thus, we have the
following bounds:
   \begin{eqnarray}
   \label{eqn-limsup-decen}
\limsup_{k \rightarrow \infty} \frac{1}{k} \log  \chi_{i,k}^{(0)} \left(
[0,+\infty) \right) \leq -\inf_{t \in[0,+\infty)} I_{(0)}(t)\\
\label{eqn-liminf-decen}
\liminf_{k \rightarrow \infty} \frac{1}{k} \log \chi_{i,k}^{(0)} \left(
(0,+\infty) \right) \geq -\inf_{t \in (0,+\infty)} I_{(0)}(t)
\end{eqnarray}
Due to the continuity of the function $I_{(0)}(\cdot)$ (see
eqn.~\eqref{eqn-rate-fcns-for-centralized-detection}),
the infima on the righthand sides in
eqns.~\eqref{eqn-limsup-decen}~and~\eqref{eqn-liminf-decen}
 are equal; it is easy to see that they are equal to $-I_{(0)}(0)$. Thus,
we have:
 \begin{eqnarray*}
 -I_{(0)}(0) &\leq& \liminf_{k \rightarrow \infty} \frac{1}{k} \log
\chi_{i,k}^{(0)} \left( (0,+\infty) \right)\nonumber \\
& \leq &
\limsup_{k \rightarrow \infty} \frac{1}{k} \log \chi_{i,k}^{(0)}
\left( (0,+\infty) \right) \nonumber \\
 &\leq& \limsup_{k \rightarrow \infty} \frac{1}{k} \log \chi_{i,k}^{(0)}
\left( [0,+\infty) \right) \nonumber \\
 &\leq& -I_{(0)}(0).
 \end{eqnarray*}
From the last set of inequalities we conclude that:
\begin{eqnarray}
\label{eqn-dokaz-lim-alpha-k}
\lim_{k \rightarrow +\infty} \frac{1}{k} \log \alpha_{i, \mathrm{dis}}(k)
&=& - I_{(0)}(0).
\end{eqnarray}
Similarly, it can be shown that:
\begin{eqnarray}
\lim_{k \rightarrow +\infty} \frac{1}{k} \log \beta_{i, \mathrm{dis}}(k)
&=& - I_{(1)}(0) \\
&=& - I_{(0)}(0) .
\end{eqnarray}
Now, consider
\begin{equation}
\label{eqn-dokaz-lim-beta-k}
P^e_{i, \mathrm{dis}}(k) = \alpha_{i,\mathrm{dis}}(k)P(H_0) +
\beta_{i,\mathrm{dis}}(k)P(H_1),
\end{equation}
for which the following inequalities hold:
\begin{eqnarray}
\label{eqns-bounding}
P^e_{i, \mathrm{dis}}(k) &\leq &  \alpha_{i,\mathrm{dis}}(k)+
\beta_{i,\mathrm{dis}}(k) \\
P^e_{i, \mathrm{dis}}(k) &\geq &  \alpha_{i,\mathrm{dis}}(k)P(H_0). \nonumber
\end{eqnarray}
By eqns.~\eqref{eqns-bounding}, we obtain:
\begin{eqnarray}
\label{eqn-krajnje-dokaz}
\limsup_{k \rightarrow \infty} \frac{1}{k} \log
P^e_{i,\,\mathrm{dis}}(k) =
\;\;\;\;\;\;\;\;\;\;\;\;\;\;\;
\;\;\;\;\;\;\;\;\;\;\;\;\;\;\;
\;\;\;\;\;\;\;\;\;\;\;\;\;\; \nonumber
\\
 \max \left\{
 \limsup_{k \rightarrow +\infty} \frac{1}{k} \log \alpha_{i,
\mathrm{dis}}(k),\,\,
 \limsup_{k \rightarrow +\infty} \frac{1}{k} \log \beta_{i,
\mathrm{dis}}(k)\right\}
\nonumber
\\
=-I_{(0)}(0)
\;\;\;\;\;\;\;\;\;\;\;\;\;\;\;
\;\;\;\;\;\;\;\;\;\;\;\;\;\;\;
\;\;\;\;\;\;\;\;\;\;\;
\;\;\;\;\;\;\;\;\;\;\;\;\;\;\;
\;\;\;\;\;\;\;\;\;
\nonumber
\\
\liminf_{k \rightarrow \infty} \frac{1}{k} \log P^e_{i,\,\mathrm{dis}}(k)
\geq
\liminf_{k \rightarrow \infty} \frac{1}{k} \log \alpha_{i,\mathrm{dis}}(k)
\;\;\;\;\;\;\;\;\;\;\;\;
\nonumber
\\
=-I_{(0)}(0),
\;\;\;\;\;\;\;\;\;\;\;\;\;\;\;
\;\;\;\;\;\;\;\;\;\;\;\;\;\;\;
\;\;\;\;\;\;\;\;\;\;\;
\;\;\;\;\;\;\;\;\;\;\;\;\;\;\;
\;\;\;\;\;\;\;\; \nonumber
\end{eqnarray}
and the claim of Corollary follows.
\end{proof}
\mypar{Remarks on Corollary~\ref{chernoff-lemma-dist}}
Corollary~\ref{chernoff-lemma-dist}
 says that, for large $k$ (i.e.,
 in the asymptotic regime,) the Bayes probability of error at each node $i$
behaves as:
 $P^e_{i,\,\mathrm{dis}}(k) \sim e^{-kI_{(0)}(0)}$. That is,
  $P^e_{i,\,\mathrm{dis}}(k)$, for large $k$,
  decays exponentially at the best possible rate, equal to the rate
$I_{(0)}(0)$ of the (asymptotically) optimal centralized detector. This
rate does not depend
  on the network connectivity, provided that the
  graph that collects all the links that are online (at least once) within finite time window (of length $B$)
   is connected (see Assumption~\ref{assumption-W(k)}.) Intuitively, an arbitrary time varying network,
   whose nodes communicate sufficiently often (within finite length time
windows,) provides sufficient information flow
    to achieve asymptotic optimality.

    We now comment on the non asymptotic
     finite time regime. To this end, remark that
$P^e_{i,\,\mathrm{dis}}(k)$ can be expressed as:
$P^e_{i,\,\mathrm{dis}}(k)
     = F_i(k) e^{-kI_{(0)}(0)}$, where $\lim_{k \rightarrow \infty}
    \frac{1}{k} \log F_i(k)=0$ (and thus, $F_i(k)$ has no effect when $k$
grows large.) The sequence $\{F_i(k)\}$
    plays a role in a finite time regime; it clearly
    depends on the network connectivity and can be, in general,
    different for different sensors. Analysis (by simulation) of the finite
    time regime is, due to the lack of space, omitted, and is pursued elsewhere.
     We briefly comment here that our numerical experience suggests that, in the
finite time regime, the sequence $F_i(k)$ does not have a very
large effect. The best distributed sensor--detector, among all $N$ sensors,
is typically close in performance to the optimal centralized detector, in the
       finite time regime also.
\section{Summary}
\label{section-conclusion}
We applied large deviations theory to analyze
 the performance of the running consensus distributed detection algorithm.
 We considered spatially correlated
Gaussian noise and time varying networks. With running consensus, the state
at each node is updated at each time step by: 1) exchanging information
with the immediate neighbors in the network; and 2) incorporating into the decision process
new local observations. We allowed the underlying network be time varying,
 provided that the graph that collects all the links that are at least once online within a finite
  time window is connected. We showed that, under spatially correlated Gaussian noise and stated network connectivity
assumptions, the running consensus asymptotically approaches the optimal centralized
detector. That is, the Bayes probability of detection error at each
sensor decays exponentially at the best achievable rate, the Chernoff
information rate.

\nopagebreak[4]
\bibliographystyle{IEEEtran}
\bibliography{IEEEabrv,LDPBibliography}
\end{document}